\documentclass[final,5p,times,twocolumn]{elsarticle}
\usepackage{amssymb}
\newcommand{\ket}[1]{|#1\rangle}
\newcommand{\bra}[1]{\langle #1|}

\newtheorem{theorem}{Theorem}
\newenvironment{proof}[1][Proof]{\noindent\textbf{#1.} }{\ \rule{0.5em}{0.5em}}

\begin{document}

\begin{frontmatter}
\title{Entanglement fidelity of the standard quantum teleportation channel  }
\author{Gang Li}
\author{Ming-Yong Ye}
\ead{myye@fjnu.edu.cn}
\author{Xiu-Min Lin}
\address{Fujian Provincial Key Laboratory of Quantum Manipulation and New Energy Materials,
\\College of Physics and Energy, Fujian Normal University, Fuzhou 350007, China}

\begin{abstract}
We consider the standard quantum teleportation protocol where a general bipartite
state is used as entanglement resource. We use the entanglement fidelity to
describe how well the standard quantum teleportation channel transmits quantum
entanglement and give a simple expression for the entanglement fidelity when
it is averaged on all input states.
\end{abstract}


\end{frontmatter}

\section{Introduction}
Nowadays the protocol of quantum teleportation \cite{1} plays an important
role in quantum information science \cite{2}. Quantum teleportation can
naturally be related to quantum channels since there are an input and an
output state involved. Mathematically, a quantum channel is a completely
positive and trace-preserving (CPTP) operator that maps an input density
operator to an output density operator, and it can be represented in an
operator-sum form \cite{2,3}.

The property of a quantum teleportation channel is dependent on both the
entanglement resource and the particular local operations and classical
communication (LOCC ) we used \cite{4,5,7}. In a realistic quantum
teleportation the sender and the receiver usually share a mixed entangled
state, instead of a maximally entangled pure state, accounted to the
decoherence. Quantum teleportation using a mixed entangled state is equivalent
to a noisy quantum channel. In 2001, it was shown that the standard quantum
teleportation protocol using a mixed entangled resource is the same as a
generalized depolarizing channel \cite{7}.

In this paper we consider the standard quantum teleportation protocol where a
general bipartite state is used as the entanglement resource. It is known that
the ordinary fidelity \cite{7b,8} between the input state and the output state is
usually used to measure the quality of a quantum teleportation channel
\cite{4,5}. However, people may be interested in how well a quantum
teleportation channel preserves quantum entanglement in the case that the particle to be teleported
is entangled with some other particle. To answer this question,
we will consider entanglement fidelity \cite{2,9} instead of the ordinary
fidelity. To our knowledge, entanglement fidelity has not yet be used to
measure the quality of quantum teleportation channel. The main result of this
paper is to give a simple expression for the entanglement fidelity of the
standard quantum teleportation channels when it is averaged on all input states.
\section{The standard quantum teleportation}
A general quantum teleportation protocol is as follows. Suppose the sender
Alice and the receiver Bob share an entangled state $\chi_{34}$, where 3 and 4
stand for the particles shared by Alice and Bob respectively, and Alice is
given another particle 1 in an unknown state $\rho_{1}$ to be teleported to
Bob. We assume each particle is associated with a $d$-dimensional Hilbert
space. To start quantum teleportation, Alice first performs a measurement on
particles 1 and 3, which is described by a collection of measurement operators
$M_{13}^{i}$ with $\sum_{i}M_{13}^{i\dagger}M_{13}^{i}=I_{13}$, where $i$
denotes measurement result. The state of Bob's particle after the measurement
will change to
\begin{equation}
\rho_{4}^{i}=\frac{1}{p_{i}}Tr_{13}\left[  \left(  M_{13}^{i}\otimes
I_{4}\right)  \left(  \rho_{1}\otimes\chi_{34}\right)  \left(  M_{13}%
^{i\dagger}\otimes I_{4}\right)  \right]
\end{equation}
if the result $i$ occurs, where
\begin{equation}
p_{i}=Tr_{134}\left[  \left(  M_{13}^{i}\otimes I_{4}\right)  \left(  \rho
_{1}\otimes\chi_{34}\right)  \left(  M_{13}^{i\dagger}\otimes I_{4}\right)
\right]
\end{equation}
is the probability of obtaining the measurement result $i$. After obtaining
the measurement result $i$, Alice tells Bob the result $i$ via a classical
channel. Then Bob applies a quantum operation $\varepsilon^{i}$, a completely
positive and trace-preserving map, to his particle. Obviously, after the
operation the state of Bob's particle is changed to $\varepsilon^{i}\left(
\rho_{4}^{i}\right)  $. Therefore, over all measurement result $i$, the final
teleported state is given by $\gamma_{4}$ $=\sum_{i}p_{i}\varepsilon
^{i}\left(  \rho_{4}^{i}\right)  $. This protocol of quantum teleportation can
be viewed as a quantum channel $\varepsilon$ which maps the input density
operator $\rho_{1}$ to the output density operator $\gamma_{4}$ \cite{10}. In
an operator-sum form the quantum teleportation channel $\varepsilon$ can be
written as \cite{10b,10c,10d}
\begin{equation}
\gamma_{4}=\varepsilon\left(  \rho_{4}\right)  =\sum_{i}A_{4}^{i}\rho_{4}%
A_{4}^{i\dagger}%
\end{equation}
with $\sum_{i}A_{4}^{i\dagger}A_{4}^{i}=I_{4}$ and $\rho_{4}$ being the same
state as $\rho_{1}$. It is obvious that the operators $A_{4}^{i}$ depend on
the entanglement resource $\chi_{34}$, the sender's measurement operators
$M_{13}^{i}$ and the receiver's corresponding CPTP maps $\varepsilon^{i}$.

In the standard quantum teleportation, the maximally entangled state%
\begin{equation}
\left\vert \Omega^{0,0}\right\rangle =\frac{1}{\sqrt{d}}\sum_{i=0}%
^{d-1}\left\vert i\right\rangle \otimes\left\vert i\right\rangle ,
\end{equation}
is usually assumed to be used as the entanglement resource to teleport a
$d$-dimensional state \cite{1}. And the sender's measurement is a generalized
Bell measurement with measurement operators $\left\{  \left\vert \Omega
^{n,m}\right\rangle \left\langle \Omega^{n,m}\right\vert \right\}  $, which
are defined as
\begin{eqnarray}
\left\vert \Omega^{n,m}\right\rangle&=&\left(  U^{n,m} \otimes I\right)
\left\vert \Omega^{0,0}\right\rangle ,\\
n,m  &=&0,1,\cdots,d-1.
\end{eqnarray}
Here the unitary operator
\begin{equation}
U^{n,m}=\sum_{j=0}^{d-1}e^{i2\pi nj/d}\ket{j}\bra{j \oplus m ~mod~ d}.
\end{equation}

When a measurement result denoted by index $\left(  n,m\right)  $ is obtained,
i.e., the state of sender's particles is mapped to state $\left\vert
\Omega^{n,m}\right\rangle $, the receiver's corresponding CPTP map is defined
to be the unitary operator $U^{n,m}$ \cite{1}. This protocol of quantum
teleportation can teleport any $d$-dimensional state perfectly. However, it
can only be viewed as a noisy quantum channel when a general $d\times d$
bipartite state $\chi$, instead of $\left\vert \Omega^{0,0}\right\rangle $, is
used.

We will consider the "standard" quantum teleportation channel using a
general entanglement resource $\chi$, where we use the same LOCC as that when
$\left\vert \Omega^{0,0}\right\rangle $ is used. This kind of the standard quantum teleportation
 can be written in the operator-sum form
\begin{equation}
\varepsilon\left(  \rho\right)  =\sum_{n,m}p_{nm}U^{n,-m}\rho U^{n,-m\dagger},
\label{tc}%
\end{equation}
where $p_{nm}=\left\langle \Omega^{n,m}\right\vert \chi\left\vert \Omega
^{n,m}\right\rangle $ and $\rho$ is the state to be teleported \cite{7}.
\section{Entanglement fidelity of the standard quantum teleportation channel}
In this section we consider the standard quantum teleportation channel
$\varepsilon$ where a general $d\times d$ bipartite state $\chi$ is used as
the entanglement source. This standard quantum teleportation channel
$\varepsilon$ has an operator-sum form as shown in Eq. (\ref{tc}). Our main
purpose is to give a quantity to measure how well quantum entanglement is
preserved by this standard quantum teleportation channel.

The standard quantum teleportation channel $\varepsilon$ in Eq. (\ref{tc})
generally cannot transport quantum state perfectly and a way is needed to
measure how well the output state $\varepsilon\left(  \rho\right)  $ is
similar to the input state $\rho$. The fidelity \cite{7b,8} of the input state
$\rho$ and the output state $\varepsilon\left(  \rho\right)  $ can be used to
do this, which is defined as%
\begin{equation}
F\left(  \rho,\varepsilon\left(  \rho\right)  \right)  =\left(  Tr\sqrt
{\rho^{\frac{1}{2}}\varepsilon\left(  \rho\right)  \rho^{\frac{1}{2}}}\right)
^{2}.
\end{equation}
When the input state is a pure state $\left\vert \psi\right\rangle $ the
fidelity will be
\begin{equation}
F\left(  \left\vert \psi\right\rangle ,\varepsilon\left(  \left\vert
\psi\right\rangle \left\langle \psi\right\vert \right)  \right)  =\left\langle
\psi\right\vert \varepsilon\left(  \left\vert \psi\right\rangle \left\langle
\psi\right\vert \right)  \left\vert \psi\right\rangle .
\end{equation}
It measures the similarity between the output state $\varepsilon\left(
\left\vert \psi\right\rangle \left\langle \psi\right\vert \right)  $ and the
input state $\left\vert \psi\right\rangle $ in the way that it will be zero
when the output state $\varepsilon\left(  \left\vert \psi\right\rangle
\left\langle \psi\right\vert \right)  $ is orthogonal to the input state
$\left\vert \psi\right\rangle $ and be the unit when the output state
$\varepsilon\left(  \left\vert \psi\right\rangle \left\langle \psi\right\vert
\right)  $ is the same as the input state $\left\vert \psi\right\rangle $.

However, there are many possible input states and different input states can
lead to different fidelities. The average of the fidelity $F\left(  \left\vert
\psi\right\rangle ,\varepsilon\left(  \left\vert \psi\right\rangle
\left\langle \psi\right\vert \right)  \right)  $ over all input pure state
$\left\vert \psi\right\rangle $ is usually introduced to characterize the
quality of the standard quantum teleportation channel $\varepsilon$.
Precisely, the quantity
\begin{equation}
\overline{F}\left(  \varepsilon\right)  =\int_{\psi}d\psi\left\langle
\psi\right\vert \varepsilon\left(  \left\vert \psi\right\rangle \left\langle
\psi\right\vert \right)  \left\vert \psi\right\rangle
\end{equation}
is used to measure how well the standard quantum teleportation channel
$\varepsilon$ is similar to a perfect channel, where the integral is performed
with respect to the uniform distribution $d\psi$ over all input pure states
\cite{4}. For the standard quantum teleportation channel $\varepsilon$ using a
general $d\times d$ bipartite state $\chi$ as the entanglement resource, it
has been shown that
\begin{equation}
\overline{F}\left(  \varepsilon\right)  =\frac{d}{d+1}f+\frac{1}{d+1},
\end{equation}
where
\begin{equation}
f=p_{00}=\left\langle \Omega^{0,0}\right\vert \chi\left\vert \Omega
^{0,0}\right\rangle \label{f}%
\end{equation}
is the generalized singlet fraction \cite{4,11,12}.

We can also make use of the entanglement fidelity \cite{2,9} to characterize
the similarity between the input $\rho$ and the output state $\varepsilon
\left(  \rho\right)  $ of the standard quantum teleportation channel
$\varepsilon$. The entanglement fidelity $F_{e}\left(  \rho,\varepsilon\left(
\rho\right)  \right)  $ is defined as%
\begin{equation}
F_{e}\left(  \rho,\varepsilon\left(  \rho\right)  \right)  =\left\langle
\varphi\right\vert I\otimes\varepsilon\left(  \left\vert \varphi\right\rangle
\left\langle \varphi\right\vert \right)  \left\vert \varphi\right\rangle,
 \label{vv}%
\end{equation}
where $\left\vert \varphi\right\rangle $ is a $d\times d$ bipartite state and
is a purification of the input state $\rho$. The entanglement fidelity
$F_{e}\left(  \rho,\varepsilon\left(  \rho\right)  \right)  $ measures how
well the entangled state $\left\vert \varphi\right\rangle $ is preserved. We
note that any purification of the input state $\rho$ can be used in Eq.
(\ref{vv}) and it always gets the same result \cite{2}. The entanglement
fidelity $F_{e}\left(  \rho,\varepsilon\left(  \rho\right)  \right)  $ is
dependent on the input state $\rho$, but we can use%
\begin{equation}
\overline{F}_{e}\left(  \varepsilon\right)  =\int_{\varphi}d\varphi
\left\langle \varphi\right\vert I\otimes\varepsilon\left(  \left\vert
\varphi\right\rangle \left\langle \varphi\right\vert \right)  \left\vert
\varphi\right\rangle \label{av}%
\end{equation}
to measure how well the standard quantum teleportation channel $\varepsilon$
preserve quantum entanglement, where the integral is performed with respect to
the uniform distribution $d\varphi$ over all $d\times d$ bipartite pure
states, which is equal to sample mixed state $\rho$ uniformly with respect
 to Hilbert-Schmidt measure \cite{12b}. Our main result is to give an expression for $\overline{F}_{e}\left(
\varepsilon\right)  $, which is summarized in the following theorem:

\begin{theorem}
In the standard quantum teleportation channel $\varepsilon$ where a general
$d\times d$ bipartite state $\chi$ is used as the entanglement resource, the
average of the entanglement fidelity $\overline{F}_{e}\left(  \varepsilon
\right)  $ defined in Eq. (\ref{av}) is given by
\begin{equation}
\overline{F}_{e}\left(  \varepsilon\right)  =\frac{d^{2}}{d^{2}+1}f+\frac
{1}{d^{2}+1},
\end{equation}
where $\ f=\left\langle \Omega^{0,0}\right\vert \chi\left\vert \Omega
^{0,0}\right\rangle $ is the generalized singlet fraction defined in Eq.
(\ref{f}).

\end{theorem}

\begin{proof}
The standard quantum teleportation channel $\varepsilon$ has an operator-sum
form as shown in Eq. (\ref{tc}), and we can submit it to $\overline{F}%
_{e}\left(  \varepsilon\right)  $ in Eq. (\ref{av}) to get%
\begin{equation}
\overline{F}_{e}\left(  \varepsilon\right)  =\sum_{n,m}p_{nm}\int_{\varphi
}d\varphi\lambda_{nm}\left(  \varphi\right)
\end{equation}
where
\begin{equation}
\lambda_{nm}\left(  \varphi\right)  =\left\vert \left\langle \varphi
\right\vert \left(  I\otimes U^{n,-m}\right)  \left\vert \varphi\right\rangle
\right\vert ^{2}.
\end{equation}
We can also write $\lambda_{nm}\left(  \varphi\right)  $ as%
\begin{equation}
\lambda_{nm}\left(  \varphi\right)  =\left(  \left\langle \varphi\right\vert
\otimes\left\langle \varphi\right\vert \right)  \mu_{nm}\left(  \left\vert
\varphi\right\rangle \otimes\left\vert \varphi\right\rangle \right)  ,
\label{qwa}%
\end{equation}
where%
\begin{equation}
\mu_{nm}=I\otimes U^{n,-m}\otimes I\otimes U^{n,-m\dagger}. \label{u}%
\end{equation}
Our next step is to compute $\int_{\varphi}d\varphi\lambda_{nm}\left(
\varphi\right)  $ using Eq. (\ref{qwa}). We first note that
\begin{equation}
 \int_{\varphi}d\varphi\lambda_{nm}\left(  \varphi\right)
 =\left\langle 00\right\vert \int_{V}\left(  V^{\dagger}\otimes V^{\dagger
}\right)  \mu_{nm}\left(  V\otimes V\right)  dV\left\vert 00\right\rangle ,
\end{equation}
where $V$ are unitary operators defined on a Hilbert space of dimension $d^{2}$ and
the integral is performed with respect to the uniform distribution $dV$ over
all unitary operators. Using Schur's lemma \cite{10,13}, we can find%
\begin{equation}
\int_{\varphi}d\varphi\lambda_{nm}\left(  \varphi\right)  =\alpha_{nm}%
+\beta_{nm},
\end{equation}
with
\begin{equation}
\alpha_{nm}=\frac{Tr\left(  \mu_{nm}\right)  }{d^{4}-1}-\frac{Tr\left(
\mu_{nm}\digamma\right)  }{d^{2}\left(  d^{4}-1\right)  }\
\end{equation}%
\begin{equation}
\beta_{nm}=\frac{Tr\left(  \mu_{nm}\digamma\right)  }{d^{4}-1}-\frac{Tr\left(
\mu_{nm}\right)  }{d^{2}\left(  d^{4}-1\right)  }%
\end{equation}
where $\digamma$ is the exchange operator. Using the identity \cite{2,14}
\begin{eqnarray}
Tr\left(  U^{n,m}\right)   &  =d\delta_{n0}\delta_{m0},\\
Tr\left(  \left(  A\otimes B\right)  \digamma\right)   &  =Tr\left(
AB\right)  .
\end{eqnarray}
We have the following results%
\begin{equation}
Tr\left(  \mu_{nm}\right)  =d^{4}\delta_{n0}\delta_{m0}, Tr\left(
\mu_{nm}\digamma\right)  =d^{2} \label{fg}%
\end{equation}
Then we have%
\begin{equation}
\alpha_{nm}=\frac{d^{4}\delta_{n0}\delta_{m0}}{d^{4}-1}-\frac{1}{d^{4}-1},
\end{equation}%
\begin{equation}
\beta_{nm}=\frac{d^{2}}{d^{4}-1}-\frac{d^{2}\delta_{n0}\delta_{m0}}{d^{4}-1}.
\end{equation}
Then we have
\begin{equation}
\overline{F}_{e}\left(  \varepsilon\right)    =\sum_{n,m}p_{nm}\left(
\alpha_{nm}+\beta_{nm}\right)
 =\frac{d^{2}}{d^{2}+1}p_{00}+\frac{1}{d^{2}+1}\nonumber
\end{equation}
Here $p_{00}=\left\langle \Omega^{0,0}\right\vert \chi\left\vert \Omega
^{0,0}\right\rangle =f$ is the generalized singlet fraction.
\end{proof}

We note that the expressions of the average $\overline{F}
_{e}\left(  \varepsilon\right)  $ of the entanglement fidelity and the average
$\overline{F}\left(  \varepsilon\right)  $ of the ordinary fidelity are very
similar, which is due to the fact that both quantities can be deduced via
Schur's lemma \cite{10,13}.
\section{Conclusion}
We use the entanglement fidelity to measure the quality of the standard quantum
teleportation channel where a general $d\times d$ bipartite state $\chi$
instead of the maximally entangled pure state $\left\vert \Omega
^{0,0}\right\rangle $ is used as the entanglement resource. We obtain an
explicit expression for the average $\overline{F}_{e}\left(  \varepsilon
\right)  $ of the entanglement fidelity, which is only dependent on the
generalized singlet fraction. Our obtained $\overline{F}_{e}\left(  \varepsilon
\right)$ quantifies how well the teleportation channel $\varepsilon$ preserves quantum entanglement.

\section*{Acknowledgments}
This work was supported by the National Natural Science Foundation of China (Grant Nos. 61275215, 11004033), the Natural Science Foundation of Fujian Province (Grant No. 2010J01002), and the National Fundamental Research Program of China (Grant No. 2011CBA00203).


\end{document}